\documentclass{article}

\usepackage{lineno,hyperref}

\usepackage{enumerate}
\usepackage{amssymb, amsmath, amsthm}
\usepackage{xcolor}
\usepackage{authblk}
\usepackage{titling}
\usepackage{graphicx}

\theoremstyle{remark}
\newtheorem*{example}{Example (TDC)}

\newtheorem{theorem}{Theorem}
\newtheorem{lemma}{Lemma}

\newtheorem{corollary}{Corollary}
\newtheorem*{assumptions*}{Assumptions}
\newtheorem{innercustomthm}{Lemma}
\newenvironment{customthm}[1]
  {\renewcommand\theinnercustomthm{#1}\innercustomthm}
  {\endinnercustomthm}

\title{Controlling the False Discovery Rate via Competition: is the +1 needed?}
\author{Andrew Rajchert\thanks{araj9923@uni.sydney.edu.au}}
\author{Uri Keich\thanks{uri.keich@sydney.edu.au}}
\affil{School of Mathematics and Statistics, The University of Sydney,\\ Camperdown, NSW 2006, Australia}
\date{}

\thanksmarkseries{arabic}

\begin{document}

\maketitle

\begin{abstract}
Barber and Cand\`es (2015) control of the FDR in feature selection relies
on estimating the FDR by the number of knockoff wins +1 divided by the number of original wins. We study the necessity of the +1 in general settings.
\end{abstract}

\section{Introduction}

In the classical multiple testing problem, we have $n$ statistical tests, each comparing a null hypothesis, $H_j$, to an alternative one.
Upon rejecting $H_j$ in favor of its alternative, we claim to make a discovery, which is a false one if in fact $H_j$ holds
(is a true null). The goal is to maximize the number of discoveries while having some control over the type I errors,
and when $n$ is large this is commonly done by controlling the false discovery rate (FDR) \cite{bh}.

The FDR is defined as the expected value of the false discovery proportion (FDP), $\text{FDP}=I/(R\vee1)$, where $R$ is the number of
reported discoveries, of which $I$ are incorrect/false, and $a\vee b=\max\{a,b\}$. Thus,
\[
\text{FDR} = E(\text{FDP}) = E\left(\frac{I}{R\vee1}\right),
\]
where the expectation is taken with respect to the true nulls.
The goal then is to maximize the number of discoveries subject to guaranteeing that the $FDR \le \alpha$, where $\alpha\in(0,1)$
is some tolerance threshold. How this problem is solved varies based on what statistical tests are being performed and what guarantees
are available on the $p$-value of each statistical test.

Barber and Cand\`es recently looked at controlling the FDR in the context of sequential hypothesis testing, where
the hypotheses are given in a prescribed order that presumably has some importance: hypotheses are typically ranked
according to some prior information on how likely are they to be rejected.
A p-value $p_j$ is associated with each $H_j$; however, unlike in the canonical setup these p-values do not have to
be particularly informative. Indeed, all that is required is that for a true null $H_j$, $p_j$, 
stochastically dominates the standard uniform distribution, $U(0,1)$: for any $u \in (0, 1)$, $P(p_j \le u) \le u$.
To establish finite-sample FDR control (explained next) they further assumed that the true null p-values are identically distributed
and independent of each other, as well as of the false null $p$-values \cite{bar}.

Barber and Cand\`es proved that the following Selective Sequential Step+ (SSS+) procedure that they introduced
rigorously controls the FDR in this setting at any predetermined level $\alpha \in (0, 1)$:
SSS+ uses a parameter $c \in (0, 1)$ and defines the rejection threshold $K_1$ as
\begin{align}
K_1 = \textstyle \max_0 \displaystyle \Big\{k \in [n] : \frac{|\{j \leq k : p_j > c\}| + 1}{ |\{j \leq k : p_j \leq c\}| \vee 1} \leq \frac{1-c}{c} \alpha\Big\}, \label{seq1}
\end{align}
where $[n] = \{1, 2, ..., n\}$ and $\max_0 (A) = \max(A)$ if $A \neq \emptyset$, otherwise $\max_0(A) = 0$ (so if the inequality in (\ref{seq1}) is never satisfied, $K_1 = 0$). SSS+ then rejects (labels as discoveries) all null hypotheses $H_j$ among the first $K_1$ hypotheses for which $p_j \le c$.

The intuition behind the definition of $K_1$ is that if the true null $p$-values are $U(0,1)$ random variables,
then $\frac{c}{1-c}|\{j \leq k : p_j > c\}|$ provides an estimated upper bound on the number of true null hypotheses $H_j$ with
$j \leq k$ and $p_j \leq c$. Therefore, dividing this expression by the corresponding number of discoveries,
$|\{j \leq k : p_j \leq c\}|\vee 1$, yields an approximate upper bound of the FDP, and hence of the FDR,
which is being kept below $\alpha$.
As we are maximizing over all $k$, the +1 in the numerator is used to counteract the bias introduced with maximizing this ratio.

In practice, the usefulness of this procedure comes from the crudeness of the $p$-values it relies on: they only need
to stochastically dominate the uniform distribution (rather than follow the $U(0,1)$ distribution), and we only ask
if they are larger or smaller than the predefined value $c$.
In particular, with $c = 1/2$, SSS+ generalizes the increasingly popular competition-based approach to controlling the FDR.
Indeed, SSS+ was introduced by Barber and Cand\`es at the same time they introduced their knockoff approach to controlling
the FDR in variable selection in a linear regression problem.
Specifically, they used an elaborate construction to pair each original variable with an artificially introduced
``knockoff'' variable so they can compete the two against each other in terms of their contribution to the linear model.
The number of knockoff wins is then used to estimate and control the FDR using SSS+ with $c=1/2$ \cite{bar}.

Exactly the same approach was used for a longer time in the analysis of tandem mass spectrometry data \cite{elias}. Often referred to as target-decoy competition (TDC), it relies on associating with each hypothesis a "target" score and a "decoy" score (analogous to the knockoff score). These scores are presumably generated so that for true null hypotheses the target and decoy are equally likely to win (have the higher score) independently of all other hypotheses. Thus, should the target score be larger than the decoy score, we set $p_j = 1/2$, otherwise $p_j = 1$. With $c = 1/2$, this satisfies the SSS+ assumptions which is then used to determine which "target wins" are taken as discoveries.

Thus, Barber and Cand\`es' proof that SSS+ controls the FDR for any finite $n$ established the
finite sample FDR-control of their knockoff+ procedure, as well TDC's (with the same +1 ``correction'',
where the latter was also established independently by He et al.~\cite{he}).
Following Barber and Cand\`es' work and the introduction of a more flexible formulation of the variable
selection problem in the model-X framework of Cand\'es et al.~\cite{candes:panning}, competition-based FDR control
has gained a lot of interest in the statistical and machine learning communities, where it
has been applied to various applications in biomedical research. However, to establish finite sample (as opposed
to asymptotic or empirical) FDR control those methods rely on SSS+, or, less frequently, on a genaralization of it called
Adaptive SeqStep~\cite{lei:power} for implementing the actual competition-based FDR control.

As SSS+ forms the basis for finite sample FDR-control, it is natural to ask whether the +1 in the numerator of (\ref{seq1})
is really necessary. Specifically, in this work we consider a variant of SSS+, which we call $SSS_t+$, where the rejection
threshold is defined in terms of a more general additive constant $t$:
\begin{align}
K_t = \textstyle \max_0 \displaystyle \Big\{k \in [n] : \frac{|\{j \leq k : p_j > c\}| + t}{|\{j \leq k : p_j \leq c\}|\vee 1} \leq \frac{1-c}{c} \alpha\Big\}. \label{Filter}
\end{align}

For any $t < 1$, we have $K_t \geq K_1$, thus applying the procedure with smaller values of $t$ could only increase its power compared to SSS+. Barber and Cand\`es showed that when using $t = 1$, the $SSS_t+$ procedure controls the FDR, i.e., $FDR(SSS_1+) \leq \alpha$ \cite{bar}. He et al.~have further showed that for any fixed $t \in (0, 1)$, there exists some $\alpha, c \in (0, 1)$ and an example of $n$ hypotheses for which $SSS_t+$ fails to control the FDR, i.e., $FDR(SSS_t+) > \alpha$ when using the altered rejection threshold given by (\ref{Filter}) \cite{he}.

One gap these results leave, which is addressed in this paper is what occurs when, as is commonly done in practice, $\alpha$ and $c$ are predetermined. That is, given $\alpha$ and $c$, is it possible to find a $t < 1$ such that $SSS_t+$ still controls the FDR, allowing for a more powerful statistical procedure?

We answer this question in two parts. First, we discuss the behaviour of $K_t$ at different values of $t$ when $\alpha$
and $c$ are fixed, from which we make conclusions about what values of $t$ will necessarily control the FDR.
Second, through the use of an explicit construction we find values of $t<1$ where the FDR is uncontrolled.

Our latter analysis is asymptotic: we prove there exists a sufficiently large $n$, such that $SSS_t+$ fails to control the FDR
when applied to our construction with $n$ hypotheses. Because in practice $n$ is also fixed in advance,
we complement the theoretical result with numerical simulations showing that for the same values of $t<1$, $SSS_t+$ 
apparently already fails to control the FDR when applied to our construction with even moderately large $n$.

\section{Equivalent Rejection Thresholds}

Let $T_k := |\{j \leq k : p_j \leq c\}|$, and $D_k := |\{j \leq k : p_j > c\}|$. In the context of TDC, these are the number of target wins and number of decoy wins in the first $k$ hypotheses respectively. With this notation our rejection threshold becomes
\begin{align}
K_t = \textstyle \max_0 \displaystyle\Big\{k \in [n] : \frac{D_k + t}{T_k\vee 1} &\leq \frac{1-c}{c}\alpha\Big\}. \label{eq2}
\end{align}

Our first theorem and its corollary show that, in many cases, it is possible to find a $t < 1$ such that $SSS_t+$ still controls the FDR, however those cases are of no practical use because $K_t = K_1$.

\begin{theorem}
\label{theorem1}
Suppose $\alpha, c \in (0, 1)$ are given such that $\frac{1-c}{c}\alpha = \frac{a}{b}$ where $a$ and $b$ are positive coprime integers. Let $m = \lceil tb \rceil$, or equivalently, $m \in \mathbb{N}$ is chosen so $t \in (\frac{m-1}{b}, \frac{m}{b}]$. Then, for any $k \geq 1$,
\begin{align}
    \frac{D_k + t}{T_k\vee1} \leq \frac{1-c}{c}\alpha \iff \frac{D_k + m/b}{T_k\vee1} \leq \frac{1-c}{c}\alpha \label{equi}
\end{align}

\end{theorem}

\begin{corollary}
\label{cor1}
If $\alpha$ and $c$ satisfy the conditions of Theorem \ref{theorem1}, then for any $t \in (1-\frac{1}{b}, 1]$, $K_t = K_1$
\end{corollary}

\begin{proof}[Proof of Corollary]
We have $m = \lceil tb \rceil = b$, so by Theorem \ref{theorem1}, $K_t = K_{m/b} = K_1$.
\end{proof}

This corollary implies that although the additive constant $t$ can be reduced in some cases while $SSS_t+$ maintains control of the FDR, this reduction has no change on the rejection threshold, and thus has no effect on the list of discoveries. To gain an increase in power, $t$ must be reduced further such that there are situations where $SSS_t+$ performs differently to SSS+.

\begin{proof}[Proof of Theorem \ref{theorem1}]
As $t \leq m/b$, the $\Longleftarrow$ implication is obvious and we are left to show that $\Longrightarrow$ holds.

Replacing $\frac{1-c}{c}\alpha$ with $a/b$ and rearranging the expression we see that the left-hand side of (\ref{equi}) is equivalent to
\begin{align*}
\frac{b(D_k + t)}{a} \leq T_k\vee1.
\end{align*}

Notice the right hand side of the above inequality is an integer, hence we may apply the ceiling function to the left hand side and preserve the inequality:
\begin{align*}
&\Big\lceil \frac{b(D_k + t)}{a} \Big\rceil \leq T_k\vee1,\\
\implies &\Big\lceil \frac{bD_k + m}{a} - \frac{m-bt}{a}\Big\rceil \leq T_k\vee1.
\end{align*}

As $bD_k + m \in \mathbb{N}$, there must exist $q, r \in \mathbb{N}$ such that $bD_k + m = aq+r$ with $0 \leq r < a$. Also note $bt \in (m-1, m]$, thus $x := m-bt \in [0, 1)$ and $(r-x)/a \in ((r-1)/a, r/a]$. If $r = 0, \lceil r/a\rceil = 0 = \lceil (r-x)/a \rceil$. Otherwise, if $r > 0$, $r \geq 1$ hence $(r-x)/a > 0$ and $r/a \leq 1$ so $\lceil r/a\rceil = 1 = \lceil (r-x)/a \rceil$. In either case, we have $\lceil r/a\rceil = \lceil (r-x)/a \rceil$ therefore
\begin{align*}
\Big\lceil \frac{bD_k + m}{a} - \frac{m-bt}{a}\Big\rceil &= \Big\lceil q + \frac{r-x}{a} \Big\rceil\\ 
&= q + \Big\lceil \frac{r}{a} \Big\rceil\\
&= \Big\lceil \frac{bD_k + m}{a} \Big\rceil.
\end{align*}

It immediately follows that
\begin{align*}
\Big\lceil \frac{bD_k + m}{a} \Big\rceil &\leq  T_k\vee1 .
\end{align*}
Removing the ceiling function and rearranging the inequality as before yields the desired inequality:
\begin{align*}
\frac{D_k + m/b}{T_k\vee1} \leq \frac{a}{b} = \frac{1-c}{c}\alpha.
\end{align*}

\end{proof}


\section{When the FDR Is Uncontrolled}

In this section we introduce a class of multiple testing problems that allows us to establish a lower bound on $t$. That is, if we wish to guarantee $SSS_t+(\alpha, c)$ controls the FDR, then we must set $t \geq t_0(\alpha, c)$. Throughout this section, we make use of one or both parts of the following assumptions on $\alpha, c$ and the true null p-values:

\begin{assumptions*}\
\begin{enumerate}
    \item Suppose $\alpha \in (0, 1)$ and $c \in (0, 1/2]$ are such that $\frac{1-c}{c}\alpha = a/b$ where $a, b \in \mathbb{N}$ are positive coprime integers such that $a < b$.
    \item The true null hypotheses p-values are independent and identically distributed independently of the false nulls with $P(p_j \leq c) = c$.
\end{enumerate}

\end{assumptions*}

Note that in TDC, $c = 1/2$ and $P(p_j\le 1/2)=1/2$ independently of all other hypotheses, so if $\alpha \in (0, 1) \cap \mathbb{Q}$,
then the above assumptions hold. We believe more generally that these assumptions are not overly restrictive and we will revisit them in the Discussion section.

\subsection{The Construction}
Our construction is periodic and determined by $\alpha$ and $c$, or more precisely, by $a$ and $b$. Note that by Assumption 1,
\begin{align*}
\gcd(a, a+b) = \gcd(b, a+b) = 1,
\end{align*}
and therefore there exist multiplicative inverses $a^{-1}$ and $b^{-1}$ of $a$ and $b$ respectively in $\mathbb{Z}_{a+b} = \mathbb{Z}/(a+b)\mathbb{Z}$. Moreover, as $-a \equiv b \pmod{a+b}$,
\begin{align}
-a^{-1} \equiv b^{-1} \pmod{a+b}. \label{e11}
\end{align}

The locations of the true null hypotheses in our construction are periodically prescribed by the following set $L$
(the remaining hypotheses are false nulls):
\begin{align}
L := \big\{k \in [n] : k \equiv -ja^{-1} \pmod{a+b} \text{ for } j \in \{0, 1, ..., 2a\}\big\}. \label{eq4}
\end{align}
That is, the positions of both the true and false null hypotheses are $a+b$ periodic and, as proven in the supplementary text, 
there are $2a+1 \leq a+b$ true nulls in each complete cycle. In particular, if $b = a+1$, there are no false nulls.
In addition, we set all the false null p-values to $p_j = c$. In general this will not necessarily be the case, however
for our purpose here of demonstrating a failure to control the FDR (the next theorem) we are at liberty
to make this choice.

\begin{example}
Suppose $c=1/2$ (as in TDC) and $\alpha=0.1$. Then $\frac{1-c}{c}\alpha = a/b$ where $a=1$ and $b=10$, 
so each complete cycle is made of $a+b=11$ hypotheses of which $2a+1=3$ are true nulls: those in
positions $\{0,-1,-2\} \equiv \{9, 10, 11\} \pmod{11}$. For example, the first cycle, which is made of the hypotheses
$H_1,\dots,H_{11}$, starts with $(a+b)-(2a+1)=b-2=8$ false nulls ($H_j$ with $j\in\{1,\dots,8\}$)
followed by $2a+1=3$ true nulls (with $j\in\{9,10,11\}$). The second cycle
cycle where $j\in\{12,\dots,22\}$ starts with false nulls at $j\in\{12,\dots,19\}$ and ends
with the true nulls at $j\in\{20,21,22\}$ etc.
A similar cycle structure of starting with $b-2a=b-2$ false nulls followed by $2a+1=3$ true null hypotheses
for a total of $a+b=b+1$ hypotheses per complete cycle applies more generally when $\alpha=1/b$ (and $c=1/2$).
\end{example}

\begin{theorem}
\label{theorem2}
Suppose that Assumptions 1 and 2 hold and let $u \in \mathbb{Z}$, $u \geq a$. Then there exists an $n \in \mathbb{N}$ such that when $SSS_t+(\alpha, c)$ is applied to our above construction with $n$ hypotheses and $t = 1-u/b$, it fails to control the FDR; in other words, $FDR(SSS_t+) > \alpha$.
\end{theorem}

\begin{corollary}
\label{cor2}
Suppose Assumptions 1 and 2 hold with $a = 1$. Then $t = 1$ is the optimal additive constant, i.e., if $t < 1$, then either $K_t = K_1$ or the rejection threshold $K_t$ will not always control the FDR.
\end{corollary}

\begin{proof}[Proof of Corollary]
By Corollary \ref{cor1}, $K_t = K_1$ for $t \in (1 - 1/b, 1]$. If $t \leq 1-1/b$, then by Theorem \ref{theorem1} we may assume
without loss of generality that $t = 1-u/b$ where $u \in \mathbb{Z}$, $u \geq 1 = a$. By Theorem \ref{theorem2}, our construction gives an example for which the FDR is not controlled at level $\alpha$. It follows that any value of $t$ that increases the power beyond $t = 1$ will not always control the FDR, thus $t = 1$ is optimal.
\end{proof}

Note that the last corollary applies to TDC for many commonly used thresholds such as
$\alpha \in \{0.01, 0.05, 0.1, 0.2\}$. It follows that for those values of $\alpha$ we cannot improve TDC
by replacing the +1 with some smaller amount. Moreover, our construction in those cases reduces to the one
described in the above example. Using the same construction we show below that although our theoretical
result is stated for an unspecified, sufficiently large number of hypotheses, $n$, using $K_t$ with $t=1-1/b$
seems to fail to control the FDR for even moderate values of $n$ (Figure \ref{fig12}).

\begin{proof}[Proof of Theorem \ref{theorem2}]
The proof of Theorem \ref{theorem2} is inspired by the proof of Theorem 2 of \cite{he} and goes through
a sequence of lemmas that is outlined below with the formal statements of the lemmas and their proofs
provided in the supplementary.

Assume the construction of (\ref{eq4}) is used and let $K = K_t(n)$ as in (\ref{eq2}).
In Lemma 1
, we show that if $K < n$, then
\begin{align}
\frac{D_K + 1}{T_K\vee1} \geq \frac{1-c}{c}\alpha. \label{e1}
\end{align}
The last inequality is useful because the entire TDC approach is based on the idea that $(D_K + 1)/(T_K\vee1)$
provides some estimate of the FDR, which in this case we are trying to show is bigger than $\alpha$.
More specifically, Lemma 3 below shows, we can leverage \eqref{e1} to get a lower bound on the FDP
when $K < n$.

As for the proof of Lemma 1, it relies mostly on the constraints our construction imposes
on the values $K$ can attain given that $K$ has to satisfy:
\[
\frac{D_K + 1-a/b}{T_K\vee1} \leq \alpha < \frac{D_{K+1} + 1-a/b}{T_{K+1}\vee1} .
\]
It follows that $p_{K+1} > c$ (decoy win in TDC), and because we set $p_j = c$ for false nulls it follows
that $K+1\in L$. 

In our TDC example above with $\alpha=0.1$ and $K=K_t$ with $t=1-1/b=0.9$ it follows
that $K\in\{8,9,10\} \pmod{11}$. A more careful analysis that is part of the general proof
shows that $K\equiv8 \pmod{11}$ is not possible. For example, $K=8$ is clearly impossible because
$(D_K+t)/T_K=0.9/8>\alpha$ in this case.
Similarly, note that $(D_{19},T_{19})\in\{(0,19),(1,18),(2,17), (3, 16)\}$ and in the first case $K\ge20$ whereas
in all other cases $(D_{19}+t)/T_{19}>\alpha$ so $K\neq19$, with a similar pattern continuing for all $K \equiv 8 \pmod{11}$. It can then be showed that \eqref{e1}
holds for the remaining $K\in\{9,10\} \pmod{11}$.

In Lemma 2
, we show that $\lim_{n \to \infty} P(K = n) = 0$. Technically we only show this along a subsequence of $m$
complete cycles, $n_m = m(a+b)$, so for the remainder of the proof we assume $n$ is of that form.

Note that in our TDC example $K=n_m=m(1+10)$ if and only if $D_n\le m-1$ and $T_n \ge 8m+(2m+1)$ (false + true null
target wins): indeed, for those bounds we have $(D_n+0.9)/T_n \le (m-0.1)/(10m+1)<\alpha$, and with $D_n \geq m$
this inequality reverses.
Thus, $K=n_m$ only if the number of decoy wins among the true nulls is less than half the number of target wins
among the same true nulls. But the probability of each is a half independently of everything else, so
the probability that $K=n_m=11m$ in this case goes to 0 as $m\to\infty$.

Combining the two lemmas, we find (\ref{e1}) holds on the events $A_m = \{K < n_m\}$ whose probability tend to 1 as $n_m \to \infty$. The significance of this is that with $I_k$ denoting the number of false discoveries among the first $k$ hypotheses, i.e., the number of $j \leq k$ such that the $j$-th hypothesis is a true null and $p_j \leq c$, the FDP among those $k$ hypotheses is given by
\begin{align*}
    Q_k := \frac{I_k}{T_k\vee1} = \frac{I_k}{D_k+1} \frac{D_k+1}{T_k\vee1}.
\end{align*}

Therefore, on the same sets $A_m$,
\begin{align*}
Q_K \geq \frac{I_K}{D_K+1} \frac{1-c}{c}\alpha.
\end{align*}

Next, in Lemma 3, we show that $Q_K$ does not converge to this lower bound in probability as we let $m$ increase to $\infty$. More specifically, we show there exist some $\epsilon, \delta > 0$ and a sequence of events $C_m$ such that for all sufficiently large $m$,
\begin{enumerate}
    \item $C_m \subset A_m$.
    \item $P(C_m) > \epsilon$.
    \item On $C_m$,
    \begin{align}
        Q_K \geq \frac{I_K}{D_K+1}\frac{1-c}{c}\alpha + \delta. \label{e3}
    \end{align}
\end{enumerate} 

It follows that with $K = K_t(n_m)$ and letting $Z_m = I_K/(D_K+1)$,
\begin{align*}
Q_K &\geq (Z_m \frac{1-c}{c}\alpha + \delta)\cdot 1_{C_m} + Z_m\frac{1-c}{c}\alpha \cdot 1_{C_m^c \cap A_m} + Q_{K}\cdot 1_{A_m^c}\\
&= \delta \cdot 1_{C_m} + Z_m \frac{1-c}{c}\alpha \cdot 1_{A_m} + Q_{K}\cdot 1_{A_m^c}
\end{align*}

Taking expectations,
\begin{align}
E(Q_K) \geq \delta \epsilon + \frac{1-c}{c}\alpha E(Z_m \cdot 1_{A_m}). \label{e4}
\end{align}

Finally, in Lemma 4
, we show that 
\begin{align}
\limsup_{m \to \infty} E(Z_m \cdot 1_{A_m}) = \frac{c}{1-c}. \label{e5}
\end{align}

It follows from (\ref{e4}) and (\ref{e5}) that
\begin{align*}
\limsup_{m \to \infty}E(Q_K) \geq \delta \epsilon + \alpha > \alpha.
\end{align*}

This establishes Theorem \ref{theorem2}.
\end{proof}

\section{Discussion}

While the conditions on $\alpha$ and $c$ in Theorem \ref{theorem2} are relatively strong, many of these hold in practice. Typically $c = \frac{1}{2}$ is used, with $c < 1/2$ being uncommon. Furthermore, $c > 1/2$ is highly unusual in practice as this will result in a reduction of statistical power as true null hypothesis tests with $p_j \leq c$ become more common, resulting in the rejection threshold tending to be small. Similarly, it is highly unusual to have irrational values of $c$ or $\alpha$ in practice, thus $\frac{1-c}{c}\alpha$ is almost always rational. Loosening these assumptions is an area for future work.

Note that under the conditions of Theorem \ref{theorem2}, Theorem \ref{theorem1} suggests that using $t > 1-1/b$ controls the FDR, while using $t = 1-u/b$ with $u \geq a$ is not guaranteed to control the FDR. If we consider the case $a = 1$, then by Corollary 2 we conclude that $t = 1$ is effectively optimal even when $\alpha$ and $c$ are given. Some frequently used values that satisfy this are $c = 1/2$, $\alpha = 0.01, 0.05, 0.1 $.
If $a > 1$ there is a non-empty interval between $1-a/b$ and $1-1/b$ and hence an uncertainty about whether the FDR is always controlled for $t \in (1-a/b, 1-1/b]$. We leave this for future investigation.

Our optimality result is a theoretical one, showing that there exists an $n$ such that $SSS_t+$ with
$t = 1-u/b$ where $u \geq a$ fails to control the FDR.
However, we complemented it using Monte Carlo simulations showing that this failure
already seems to occur for moderately large values of $n$ (Figure \ref{fig12}, $n\approx 300$), which are orders of
magnitudes smaller than the number of hypotheses we encounter when analyzing a typical tandem mass spectrometry data.

\begin{figure}[h]
    \centering
\begin{tabular}{ll}
    \includegraphics[width = 6cm]{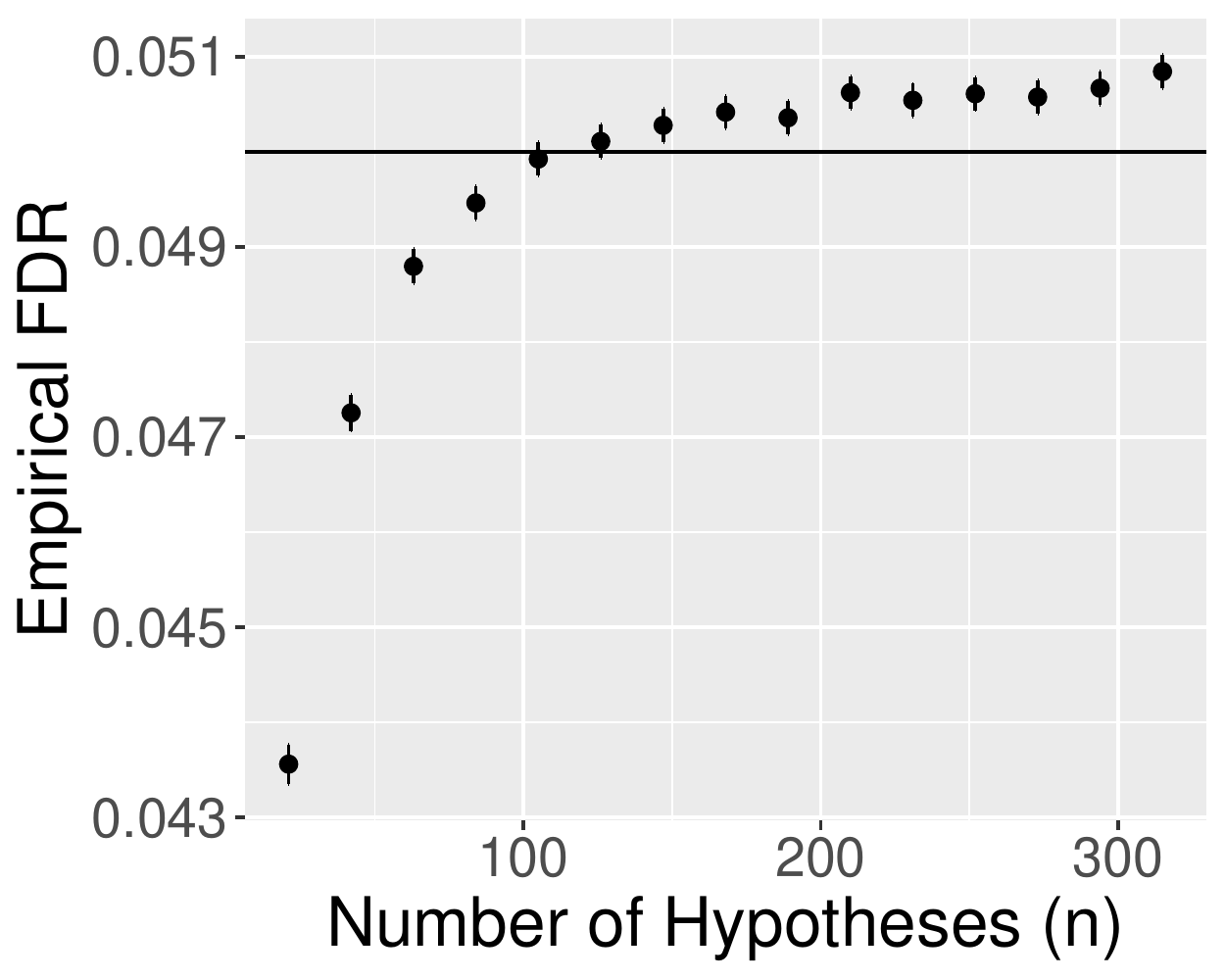} & 
    \includegraphics[width = 6cm]{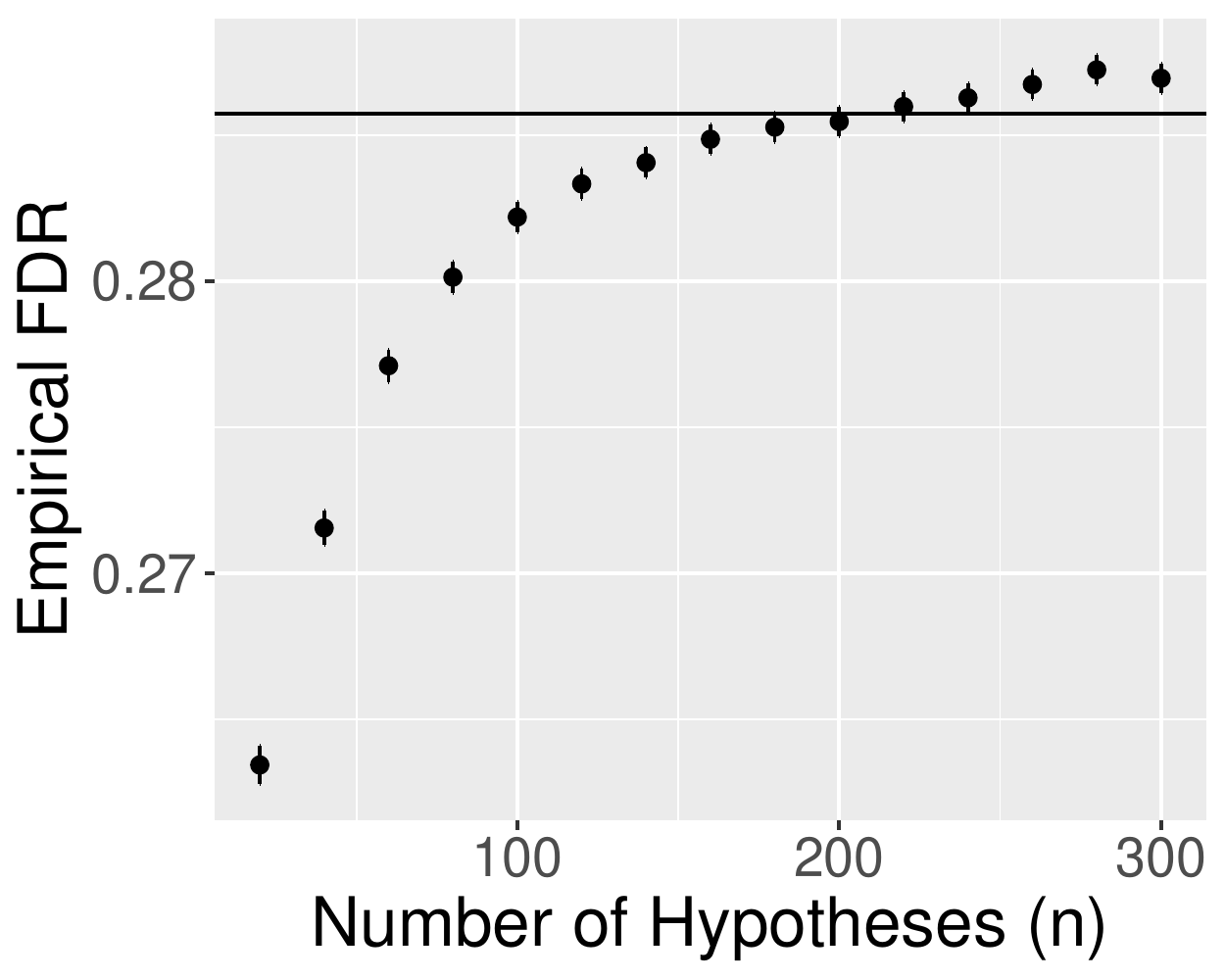}
\end{tabular}
    \caption{Empirical FDR using Monte Carlo generated data: plotted are the estimated/empirical FDR (FDP average) as well as
    99\% confidence intervals for the FDR when applying $SSS_t+$ to our construction with
    $\alpha = 0.05$, $c = 0.5$, $t = 0.95$ (left) and 
    $\alpha = 2/7$, $c = 0.4$, $t = 4/7$ (right). Number of simulations is $N = 400,000$ for each set of parameters.}
\label{fig12}
\end{figure}

Our proof of Theorem \ref{theorem2} provides a carefully constructed example that maximises a heuristic strongly associated with the FDR, and thus is among the strongest possible examples that maximise the FDR. The methodology used to produce such an example can be generalised to other applications, and is likely to be of further use in practice when investigating FDR controlling procedures related to TDC.

Finally, note that while the construction in Theorem \ref{theorem2} was designed to be optimal in the case $c = 1/2$, we do not expect it to be the case for other values of $c$. In particular, we have experimented with alternate constructions for $c < 1/2$ that allowed us to reduce the aforementioned uncertainty gap of $t \in (1-a/b, 1-1/b]$.

\newpage

\section*{Supplementary}

\begin{table}[h]
\begin{tabular}{|l|l|}
\hline
\textbf{Notation} & \textbf{Definition}                                       \\
$[n]$             & The set of positive integers less than or equal to $n$.     \\
$\max_0(A)$        & If $A$ is non-empty, the maximum element of $A$, otherwise $0$. \\
$x \vee y$        & The maximum of the set $\{x, y\}$                         \\
$c$             & A tuning parameter of SSS+ (set to 1/2 for TDC and knockoff+)\\
$\alpha$        & FDR threshold \\
$a,b$           & Natural numbers such that $\frac{1-c}{c}\alpha = \frac{a}{b}$\\
$T_k$           & The number of top $k$ hypotheses for which the p-value is $\le c$ (``target wins'')\\
$D_k$           & Same as $T_k$ but the p-value is $> c$ (``decoy wins'')\\
$I_k$             & The number of false discoveries in the $k$ top scoring hypotheses  \\
$N_k$           & The number of true nulls in the $k$ top scoring hypotheses  \\
$K=K_t$         & Discovery cutoff: all target wins in top $K$ hypotheses are reported (rejected nulls)\\
$L$             & The indices of the true nulls in our construction\\
\hline
\end{tabular}
\caption{Notations we use}
\end{table}

\section*{Lemmas used in proving Theorem 2 of the main text}

\begin{lemma}
\label{lemma1}
Suppose Assumption 1 holds and consider our above construction with $t = 1-u/b$ where $u \in \mathbb{Z}, u \geq a$ and $K = K_t$ as in (3) of the main text
. If $K < n$, then
\begin{align*}
\frac{D_K + 1}{T_K\vee 1} \geq \frac{1-c}{c}\alpha = \frac{a}{b}.
\end{align*}
\end{lemma}

\begin{proof}
Note that if $T_K = 0$, then $D_K = K$ and the lemma is obvious using the assumption that $a/b < 1$. Now consider the case $T_K > 0$ and thus $0 < K < n$. We claim that $p_{K + 1} > c$ and thus $K + 1 \in L$. To see this, we use the definition of $K$ as a maximum to give
\begin{align}
\frac{D_K + t}{T_K \vee 1} \leq \frac{a}{b} \quad \text{ and } \quad \frac{D_{K+1} + t}{T_{K+1}\vee 1} > \frac{a}{b}. \label{Kmaxineq}
\end{align}
This immediately implies $D_{K+1} = D_K + 1$ and $T_{K+1} = T_K$, thus by definition of $D_k$ we must have $p_{K+1} > c$. As only true nulls have a p-value larger than $c$ by our construction, the $K+1$-th hypothesis has to be a true null, so $K+1 \in L$.

Using the fact that $T_K + D_K = K$ and $T_K > 0$, we must have
\begin{align*}
\frac{K - T_K + 1 - u/b}{T_K} \leq \frac{a}{b} \quad \text{ and } \quad \frac{K+1 - T_{K+1} + 1 - u/b}{T_{K+1}} > \frac{a}{b}.
\end{align*}
Rearranging these inequalities using $T_K = T_{K+1}$ gives
\begin{align}
\frac{b(K+1)-u}{a+b} \leq T_K \quad \text{ and } \quad \frac{b(K+2)-u}{a+b} > T_K. \label{KTKineq}
\end{align}

To finish the proof we show below that 
\begin{align}
\label{ineq13}
\frac{b(K+1)}{a+b} \geq T_K. 
\end{align}
Rearranging the terms of the last inequality it follows that
\[
\frac{K+1}{T_K} \geq \frac{a}{b}+1,
\]
and hence using $K = T_K + D_K$ we have
\begin{align*}
\frac{D_K + 1}{T_K} \geq \frac{a}{b}.
\end{align*}
Finally, recall we assumed that $T_K > 0$ for this case, so we may replace $T_K$ with $T_K \vee 1$ establishing the lemma in this case as well.

Returning to \eqref{ineq13}, note first that if $u \geq b$ then the second inequality from (\ref{KTKineq}) immediately establishes \eqref{ineq13}.
Otherwise, $a \leq u \leq b-1$, and as $K+1 \in L$,
\begin{align}
K + 2 \equiv 1 + jb^{-1} \pmod{a+b} \text{ for some } j \in \{0, 1, ..., 2a\}, \label{e20}
\end{align}
or equivalently,
\begin{align}
b(K+2)-u \equiv b-u+j \pmod{a+b} \text{ for some } j \in \{0, 1, ..., 2a\}. \label{e12}
\end{align}
As $j$ varies from 0 to $2a$, the above RHS consists of
\begin{align}
\label{seq_ineq}
0 < b-u < b-u+1 < b-u+2 < ... < b-u+2a.
\end{align}

If we now further assume that $u > a$, then $b-u+2a \leq a+b-1$ and it follows from \eqref{e12} and \eqref{seq_ineq} that there exists an integer $p$ such that $p(a+b) + b-u \leq b(K+2)-u < (p+1)(a+b)$.

Otherwise, $u=a$ in which case we claim that the largest term in \eqref{seq_ineq} is not attainable.
Indeed, that term corresponds to $j=2a$ in \eqref{e20} and \eqref{e12}, so it is only attainable if
$K+2 \equiv 1+2ab^{-1} \pmod{a+b}$. However, noting that
 $ab^{-1} \equiv -1 \pmod{a+b}$, we find in this case that $K \equiv -3 \pmod{a+b}$.
Write $K = p(a+b)-3$ for some positive integer $p$, then (\ref{KTKineq}) gives
\begin{align*}
bp-1-\frac{b}{a+b} \leq T_K \quad \text{ and } \quad bp-1 > T_K.
\end{align*}
As $bp-1$ is an integer and $b/(a+b) < 1$, there cannot be an integer solution for $T_K$ which
satisfies both of these inequalities. 

As the largest term in \eqref{seq_ineq} is not attainable when $u=a$, and the second largest term is smaller than $a+b$ again we find from \eqref{e12} and \eqref{seq_ineq} that there exists an integer $p$ such that $p(a+b) + b-u \leq b(K+2)-u < (p+1)(a+b)$.

Hence, considering the possible values of $b(K+2)-u$ in $\mathbb{Z}_{a+b}$ and projecting elements into $\mathbb{R}/\mathbb{Z}$,
\begin{align*}
\frac{b(K+2) - u}{a+b} \pmod{1} \in \Big[\frac{b-u}{a+b}, 1\Big) \subset \mathbb{R}/\mathbb{Z}.
\end{align*}

This expression is the same as that of (\ref{KTKineq}), thus by considering that $T_K$ is an integer, we may improve the upper bound for $T_K$ from (\ref{KTKineq}):
\begin{align*}
\frac{b(K+1)}{a+b} = \frac{b(K+2)-u}{a+b} - \frac{b-u}{a+b} \geq T_K ,
\end{align*}
thus establishing \eqref{ineq13} in this case as well.
\end{proof}

In Lemma \ref{lemma2} below, we show that the probability the condition of Lemma \ref{lemma1}, $K < n$, holds increases to 1 along the subsequence $n_m = m(a+b)$. For the remainder of the proofs, we assume $n$ is of this form and make the following observation.

Let $B_k = \{k(a+b)+j : j \in [a+b] = \{1, 2, ..., a+b\}\}$. Under the assumptions of Lemma \ref{lemma1}, we claim $|B_k \cap L| = 2a+1$. As the $B_k$s form an $(a+b)$-periodic partition of the positive integers and $L$ is defined periodically with the same period, it suffices to show this holds for $k = 0$. Clearly,
\begin{align*}
B_0 \cap L = \big\{\langle kb^{-1} \rangle : k \in \{0, 1, ..., 2a\}\big\},
\end{align*}
where $\langle kb^{-1}\rangle$ is the unique integer between 1 and $a+b$ such that $\langle kb^{-1} \rangle \equiv kb^{-1} \pmod{a+b}$. Suppose $\langle k_1b^{-1}\rangle = \langle k_2b^{-1}\rangle$ for $k_1, k_2 \in 
\{0, 1, ..., 2a\}$. Then $k_1b^{-1} \equiv k_2b^{-1} \pmod{a+b}$ and thus $k_1 \equiv k_2 \pmod{a+b}$. Since $2a \leq a+b-1$, then by our restrictions on $k_1$ and $k_2$, we must have $k_1 = k_2$. It follows that $B_0 \cap L$ consists of $2a+1$ distinct elements.

\begin{lemma}
\label{lemma2}
Let $n = n_m = m(a+b)$ with $m \in \mathbb{N}$ and let $K = K_m(t)$. Under Assumptions 1 and 2, $P(K = n) \to 0$ as $m \to \infty$.
\end{lemma}

\begin{proof}
Let $X_m$ be defined as
\begin{align*}
X_m = \frac{D_n + t}{T_n + 1} = \frac{D_n + t}{n+1-D_n} = \frac{D_n/k_m + t/k_m}{(n+1)/k_m - D_n/k_m},
\end{align*}
where $k_m = m(2a+1)$. By our previous observation, $D_n \sim$ binomial($(2a+1)m, 1-c$), hence by the strong law of large numbers, $X_m$ converges almost surely
\begin{align*}
X_m \underset{m \to \infty}{\to} \lambda := \frac{1-c}{(a+b)/(2a+1) - (1-c)}.
\end{align*}

As $1-c \geq 1/2$, we find
\begin{align*}
\lambda \geq \frac{1/2}{(a+b)/(2a+1) - 1/2} = \frac{a+1/2}{b-1/2} > \frac{a}{b}.
\end{align*}

By using the definition of $K$ as a maximum and the fact that $T_n \vee 1 < 1 + T_n$, we find
\begin{align*}
P(K < n) = P\Big(\frac{D_n + t}{ T_n \vee 1} > \frac{1-c}{c}\alpha\Big) \geq P\Big(X_m > \frac{1-c}{c}\alpha\Big).
\end{align*}

As $X_m$ converges to $\lambda > \frac{1-c}{c}\alpha$ almost surely, it follows $P(K < n) \to 1$ as $m \to \infty$, which proves the lemma.
\end{proof}

\begin{lemma}
\label{lemma3}
There exists $\epsilon, \delta > 0$ and a sequence of sets $C_m$ such that for all sufficiently large $m$:
\begin{enumerate}
    \item $C_m \subset A_m$.
    \item $P(C_m) > \epsilon$.
    \item On $C_m$,
    \begin{align}
        Q_K \geq \frac{I_K}{D_K+1}\frac{1-c}{c}\alpha + \delta. \label{e3a}
    \end{align}
\end{enumerate}
\end{lemma}

\begin{proof}
Let $N_k$ denote the number of true null hypotheses at or before index $k$. We claim that there exists a $k_0 \in \mathbb{N}$ such that if $n \geq k_0$ and $D_k \geq (a+1/4)/(2a+1)N_k$ for all $k \geq k_0$, then $K < k_0$. To show this, observe that $N_{m(a+b)} = m(2a+1)$, thus we must have
\begin{align*}
N_k &\geq N_{\lfloor k/(a+b)\rfloor (a+b)}\\
&= \big\lfloor \frac{k}{a+b}\big\rfloor (2a+1) > \big(\frac{k}{a+b} - 1\big)(2a+1).
\end{align*}

Therefore assuming $D_k \geq (a+1/4)/(2a+1)N_k$ for all $k \geq k_0$, with
\begin{align*}
d = d(k) = \frac{a+1/4}{2a+1}\big(\frac{k}{a+b} - 1\big)(2a+1) = k\big(a + \frac{1}{4}\big)\big(\frac{1}{a+b}-\frac{1}{k}\big),
\end{align*}
we have $D_k > d$ and $T_k = k-D_k < k-d$. If $T_k = 0$ for some value of $k \geq 1$, then $D_k = k \geq 1$ and thus $(D_k+t)/(T_k \vee 1) > 1 \geq a/b$. Otherwise, 
\begin{align*}
\frac{D_k+t}{T_k \vee 1} &= \frac{D_k+t}{T_k} > \frac{d+t}{k-d} = \frac{d/k + t/k}{1-d/k}.
\end{align*}

As $k \to \infty$ this lower bound will converge to 
\begin{align*}
\frac{(a+1/4)\frac{1}{a+b}}{1-(a+1/4)\frac{1}{a+b}} = \frac{a+1/4}{a+b-a-1/4} > \frac{a}{b}.
\end{align*}

It immediately follows that there exists a $k_0 \in \mathbb{N}$ such that if $D_k \geq (a+1/4)/(2a+1)N_k$ for all $k \geq k_0$, then $(D_k+t)/(T_k \vee 1) > a/b = \frac{1-c}{c}\alpha$ for all $k \geq k_0$ and thus $K < k_0$.

Note that we may assume $k_0 = m_0(a+b)$ for some $m_0 \in \mathbb{N}$ where $m_0 \geq 3$. Continuing with the proof of the lemma, let $C_m$ denote the event defined by:
\begin{enumerate}[(I)]
    \item If $k \leq a+b, p_k \leq c$ (so, using TDC terminology, all the hypotheses in the first period of the construction are target wins).
    \item If $k \equiv 0 \pmod{a+b}$ and $k \leq k_0$ then $p_k \leq c$.
    \item If $k \in (a+b, k_0) \cap L$ and $k \not\equiv 0 \pmod{a+b}$, then $p_k > c$.
    \item If $k \in [k_0, n]$, the $p_k$ are such that $D_k \geq (a+1/4)/(2a+1)N_k$ for all $k \in [k_0, n]$.
\end{enumerate}

Note that it is not clear that (IV) is achievable from our construction while obeying (I) - (III), however this indeed is the case as we will show towards the end of the proof.

It follows from (IV) and our previous analysis that $K < k_0$, establishing $C_m \subset A_m$, which is statement 1 of the lemma. Furthermore by (I), if we consider $k = a+b$,
\begin{align*}
\frac{D_{k} + t}{T_{k} \vee 1} = \frac{1-u/b}{a+b} \leq \frac{a}{b}.
\end{align*}
Since $K$ by definition is the maximum value of $k$ such that the above inequality is satisfied, It follows that $K$ is at least $a + b$.

In proving Lemma \ref{lemma1}, we noted that when $0 < K < n$ we must have $p_{K+1} > c$ and $K+1 \in L$. Moreover, combining the observation that $p_{K+1} > c$ with $K < k_0$ and (II), we further note that $K+1 \not\equiv 0 \pmod{a+b}$. Hence, analogously to (\ref{e20}) and (\ref{e12}),
\begin{align*}
K + 2 \equiv 1 + jb^{-1} \pmod{a+b} \text{ for some } j \in \{1, 2, ..., 2a\},
\end{align*}
or equivalently,
\begin{align*}
b(K+2) - u \equiv b-u + j \pmod{a+b} \text{ for some } j \in \{1, 2, ..., 2a\}.
\end{align*}

Continuing along the same path of logic from Lemma \ref{lemma1} and using the fact that $j \neq 2a$ when $u = a$ implies
\begin{align*}
\frac{b(K+2) - u}{a+b} \pmod{1} \in \Big[\frac{b-u+1}{a+b}, 1\Big) \subset \mathbb{R}/\mathbb{Z},
\end{align*}
and it follows that
\begin{align*}
\frac{b(K+1)}{a+b} > \frac{b(K+2)-u}{a+b}-\frac{b-u+1}{a+b} \geq \Big\lfloor \frac{b(K+2)-u}{a+b}\Big\rfloor.
\end{align*}

Using (\ref{KTKineq}) again, we have
\begin{align*}
\frac{b(K+1)}{a+b} > T_K,
\end{align*}
and rearranging this while using $T_K + D_K = K$, $K > 0$ and $p_1 \leq c$ gives
\begin{align*}
\frac{D_K+1}{T_K \vee 1} = \frac{D_K+1}{T_K} > \frac{a}{b}.
\end{align*}

As $K \geq a+b$, $I_K \geq 2a+1 > 0$ and it follows that
\begin{align*}
\frac{I_K}{T_K \vee 1} > \frac{1-c}{c}\alpha \frac{I_K}{D_K+1}.
\end{align*}

Since $K < k_0$ and $k_0$ is independent of $n$, there are finitely many values $I_K, D_K$ and $T_K$ can take, hence there exists a $\delta > 0$ such that 
\begin{align*}
\frac{I_K}{T_K \vee 1} > \frac{1-c}{c}\alpha \frac{I_K}{D_K+1} + \delta,
\end{align*}
thus establishing statement 3 of the lemma.

Finally, to prove statement 2 holds, recall that $k_0 = m_0(a+b)$ with $m_0 \geq 3$. In particular, $m_0 \geq 3 > 2a/(a-1/4)$ and therefore $2a(m_0-1) > (a+1/4)m_0$. By the definition of $C_m$, $D_{k_0} = 2a(m_0-1)$ and by construction $N_{k_0} = m_0(2a+1)$, hence
\begin{align*}
D_{k_0} > (a+1/4)m_0 = (a+1/4)\frac{N_{k_0}}{2a+1} = \gamma N_{k_0},
\end{align*}
where $\gamma = (a+1/4)/(2a+1)$.

The last inequality guarantees that (IV) is achievable with $n = k_0 = m_0(a+b)$ and therefore $P(C_{m_0}) = (1-c)^{2a(m_0-1)}c^{a(2+m_0)} > 0$. We finish the proof by showing that $\inf_{m \geq m_0} P(C_m | C_{m_0}) > 0$.

Indeed, consider extending our construction to an infinite sequence of hypotheses, and let $C_{\infty} = \{D_k \geq \gamma N_k \text{ for all } k \geq k_0\} \cap C_{m_0}$. Clearly $C_\infty \subset C_m$ for all $m \geq m_0$ hence it suffices to show that $P(C_\infty | C_{m_0}) > 0$.

Note that
\begin{align*}
P(C_\infty | C_{m_0}) = P\big[D_k \geq \gamma N_k, \quad \forall k \geq k_0 | D_{k_0} = 2a(m_0+1)\big],
\end{align*}
and because $D_k$ is a sum of $N_k$ i.i.d. Bernoulli($1-c$) random variables, the following lemma applied with $p = 1-c \geq 1/2 > \gamma$ and $s_0 = D_{k_0} = 2a(m_0+1) \geq \gamma N_{k_0}$ completes the proof.

\end{proof}

\begin{customthm}{5}
Consider an infinite sequence of i.i.d. Bernoulli($p$) random variables, and let $S_l$ denote the sum of the first $l$ random variables. Fix $\gamma \in (0, p)$ and $l_0 \in \mathbb{N}$. Then, with $s_0 \geq \gamma l_0$,
\begin{align*}
P(S_l \geq \gamma l \text{ for all } l \geq l_0 | S_{l_0} = s_0) > 0.
\end{align*}
\end{customthm}

\begin{proof}
Consider the sequence of events $B_n = \{S_l \geq \gamma l \text{ for all } l \geq n\}$ and the event $B = \{S_{l_0} = s_0\}$. We wish to show $P(B_{l_0} | B) > 0$. As $B_n$ is a sequence of increasing events,
\begin{align*}
\lim_{n \to \infty} P(B_n | B) &= P(\cup_{n=l_0}^\infty B_n | B)\\
&\geq P(\liminf_{n \to \infty} \frac{S_n}{n} \geq \gamma' | B),
\end{align*}
where $\gamma' \in (\gamma, p)$.

For $n > l_0$, $S_n - S_{l_0}$ is a Binomial($n-l_0, p$) random variable, hence applying the strong law of large numbers, we find that even conditional on $B$, $S_n/n$ converges in probability to $p$. As $\gamma' < p$, we conclude
\begin{align*}
\lim_{n \to \infty} P(B_n | B) = 1.
\end{align*}

It follows there exists some $n_0 \geq l_0$ such that $P(B_{n_0} | B) > 0$. By the law of total probability and the clear independence of $B$ and the Bernoulli random variables after index $l_0$,
\begin{align*}
P(B_{n_0} |B) = \sum_{n = 0}^{n_0-l_0} P(B_{n_0} | S_{n_0} - S_{l_0} = n, B)P(S_{n_0} - S_{l_0} = n) > 0.
\end{align*}

Thus, at least one term in the sum must be positive, and as $P(B_{n_0} | S_{n_0} - S_{l_0} = n, B)$ is increasing over $n$ with $P(S_{n_0} - S_{l_0} = n)$ never being zero over the summation index, the final term of the sum must is positive. That is, $P(B_{n_0}, S_{n_0}-S_{l_0} = n_0-l_0 | B) > 0$, and therefore with $A = B_{n_0} \cap B \cap \{S_{n_0} - S_{l_0} = n_0-l_0\}$, $P(A) > 0$.

However, the event $A$ is a subset of $B_{l_0} \cap B$: indeed, if $S_{n_0} - S_{l_0} = n_0-l_0$ then for $l_0 \leq l < n_0$, $S_l - S_{l_0} = l-l_0$ and therefore if $A$ occurs than for $l \in [l_0, n_0]$,
\begin{align*}
\frac{S_l}{l} &= \frac{S_l - S_{l_0}}{l} + \frac{S_{l_0}}{l}\\
&\geq \frac{l-l_0}{l} + \gamma \frac{l_0}{l} > \gamma,
\end{align*}
showing that $B_{l_0} \cap B$ also occurs. It follows that $P(B_{l_0} \cap B) > 0$ and we conclude $P(B_{l_0} | B) > 0$ as required.
\end{proof}

\begin{lemma}
\label{lemma4}
Suppose Assumptions 1 and 2 hold and let $A_m = \{K < n\}$ where $n = n_m = m(a+b)$ and $K = K_t(n)$. Then
\begin{align*}
\limsup_{m \to \infty} E\Big(\frac{I_K}{D_K+1} \cdot 1_{A_m}\Big) = \frac{c}{1-c}
\end{align*}
\end{lemma}

\begin{proof}
Let $X_n$ denote the number of true nulls with $p_j \leq c$ before the first true null with $p_j > c$, i.e., the number of true null target wins before the first true null decoy win. If there are no such decoy wins then let $X_n = I_n$:
\begin{align*}
X_n = |\{ j \in L : p_j \leq c, j < \min\{ i : p_i > c \}\}|.
\end{align*}

Considering only the $N_K = D_K + I_K$ true nulls among the first $K$ hypotheses, each arrangement of the decoy and target wins is equally likely for a given number of decoy and target wins among the true nulls.

Consider the locations of the $D_K$ decoy wins as defining $D_K + 1$ bins, some possibly empty, and let $J_K^i$ be the number of target null wins that fall in the $i$-th bin for $i = 1, ..., D_K+1$. That is, $J_K^i$ is the number of null target wins whose score places them between the $(i-1)$-th and $i$-th null decoy wins (where the 0th and $D_K+1$-th decoy wins refer to the start and end of the combined list of null hypotheses).

The aforementioned symmetry implies that conditioning on $D_K$ and $I_K$, $E(J_K^i | D_K, I_I)$ should be the same for each $i$, and of course $\sum_{i=1}^{D_K+1} J_K^i = I_K$. It follows that $E(J_K^i | D_K, I_K) = I_K/(D_K+1)$. Additionally, also observe that $J_K^1 = X_n$. This is trivial in the cases that $D_K > 0$, as well as $D_K = 0$ and  $K = n$. If $D_K = 0$ and $K < n$, then the inequalities from from (\ref{Kmaxineq}) apply and thus we again have $p_{K+1} \in L$ so again we have $J_K^1 = X_n$.

Note that $1_{K < n} \in \sigma(I_K, D_K)$, thus it follows that
\begin{align*}
E(X_n \cdot 1_{K < n} | I_K, D_K) = \frac{I_K}{D_K+1} \cdot 1_{K < n},
\end{align*}
and in particular,
\begin{align*}
E(X_n \cdot 1_{K < n}) = E\Big(\frac{I_K}{D_K+1} \cdot 1_{K < n}\Big). 
\end{align*}

Now consider the limiting infinite sequence of true nulls and let $X$ denote the number of true null target wins before the first decoy win. Clearly $X$ is a geometric($1-c$) random variable and $X_n$ increases to $X$ almost surely. By Lemma \ref{lemma2} and the Borel-Cantelli Lemma there exists a subsequence $n_{m'}$ along which $1_{K < n} \to 1$ as $m' \to \infty$ almost surely. It follows by dominated convergence
\begin{align*}
\limsup_{m \to \infty} E\Big(\frac{I_K}{D_K + 1} \cdot 1_{K < n}\Big) = \limsup_{m \to \infty} E(X_n \cdot 1_{K < n_m}) = E(X) = \frac{c}{1-c}.
\end{align*}
\end{proof}

\bibliographystyle{plain}
\bibliography{bibliography.bib}

\end{document}